\documentclass[journal, 12pt, draftcls, onecolumn]{IEEEtran}
\usepackage[top=1in, bottom=1in, left=1in, right=1in]{geometry}
\usepackage{cite}
\usepackage{footnote}
\usepackage{amsmath,amsfonts,amssymb,graphicx}
\usepackage{graphicx}
\usepackage{ifpdf}
\usepackage{amsmath}
\usepackage{amsthm}
\newtheorem{theorem}{Theorem}

\ifCLASSINFOpdf
\else
\fi

\begin{document}

\bibliographystyle{IEEEtran}
\title{Optimal Power Allocation for Parameter Tracking in a Distributed Amplify-and-Forward Sensor Network}
\author{\IEEEauthorblockN{Feng Jiang, Jie Chen, and A. Lee Swindlehurst}\\
\IEEEauthorblockA{Center for Pervasive Communications and Computing\\University of California at Irvine\\Irvine, CA 92697, USA\\Email:\{feng.jiang, jie.chen, swindle\}@uci.edu}\thanks{This work was supported by the Air
Force Office of Scientific Research grant FA9550-10-1-0310, and by the
National Science Foundation under grant CCF-0916073.}}
\maketitle
\begin{abstract}
We consider the problem of optimal power allocation in a sensor
network where the sensors observe a dynamic parameter in noise and
coherently amplify and forward their observations to a fusion center
(FC).  The FC uses the observations in a Kalman filter to track the
parameter, and we show how to find the optimal gain and phase of the
sensor transmissions under both global and individual power constraints
in order to minimize the mean squared error (MSE) of the parameter estimate.
For the case of a global power constraint, a closed-form solution can
be obtained.  A numerical optimization is required for individual power
constraints, but the problem can be relaxed to a semidefinite programming
problem (SDP), and we show that the optimal result can be constructed from the
SDP solution.  We also study the dual problem of minimizing global and
individual power consumption under a constraint on the MSE.  As before,
a closed-form solution can be found when minimizing total power, while
the optimal solution is constructed from the output of an SDP when minimizing
the maximum individual sensor power.  For purposes of comparison, we derive
an exact expression for the outage probability on the MSE for equal-power
transmission, which can serve as an upper bound for the case of optimal
power control.  Finally, we present the results of several
simulations to show that the use of optimal power control provides a
significant reduction in either MSE or transmit power compared with a
non-optimized approach (i.e., equal power transmission).

\begin{keywords}
Distributed estimation, distributed tracking, wireless sensor networks, amplify-and-forward networks
\end{keywords}
\end{abstract}

%\newpage
\section{Introduction}
\subsection{Background}
In a distributed analog amplify-and-forward sensor network, the sensor
nodes multiply their noisy observations by a complex factor and transmit
the result to a fusion center (FC).  In a coherent multiple access
channel (MAC), the FC uses the coherent sum of the
received signals to estimate the parameter. It has been shown that for
Gaussian sensor networks, an analog transmission scheme such as this can
achieve the minimum distortion between the source and the recovered signal
\cite{Gastpar:2003, Gastpar:2005, Gastpar:2008}. The key problem in this
setting is designing the multiplication factor for each sensor to meet
some goal in terms of estimation accuracy or power consumption.  Furthermore,
for an optimal solution, these multipliers would have to be updated in scenarios
where the parameter or wireless channels are time-varying.  In this paper,
we focus on tracking a dynamic parameter in a coherent MAC setting.

Most prior work on estimation in distributed amplify-and-forward sensor networks has focused on the situation where the parameter(s) of interest are time-invariant, and either deterministic or i.i.d. Gaussian.  The case of an orthogonal MAC, where the FC has access to the individual signals from each sensor, has been studied in \cite{Cui:2007, Xiao:2008, Bahceci:2008, Matamoros:2011, Fang:2009, Leong:2010, Senol:2008}. For a coherent MAC, relevant work includes \cite{Xiao:2008, Wang:2011, Kar:2012, jiang:20122, Banavar:2010, Leong:2010}. In \cite{Cui:2007, Xiao:2008, Bahceci:2008, Fang:2009}, two kinds of problems were considered: minimizing the estimation error under sum or individual power constraints, and minimizing the sum transmit power under a constraint on the estimation error. Scaling laws for the estimation error with respect to the number of sensors were derived in \cite{Matamoros:2011, Leong:2010} under different access schemes and for different power allocation strategies.  In \cite{Banavar:2010, jiang:20122}, the authors exploited a multi-antenna FC to minimize the estimation error.
\iffalse In \cite{Xiao:2008}, a linear minimum MSE estimator is adopted at the FC to estimate a Gaussian vector source, and the optimal power allocation problem was solved under a total
transmit power constraint. In \cite{Wang:2011, Kar:2012}, the authors also consider a coherent MAC channel. In \cite{Wang:2011}, the
outage probability of MSE is minimized under the constraint of the average sum transmit power. In \cite{Kar:2012}, the sensor nodes can collaborate with other nodes or transmit directly to the FC and to minimize the MSE, optimal collaboration strategy was found under the global or individual power constraint. A phase-only optimization problem was
formulated in \cite{jiang:20122} and the phase of the transmitted
signal from different sensor nodes was adjusted such that the received
signal at the FC can be added coherently to optimize the performance
of a maximum likelihood (ML) estimator.\fi

More relevant to this paper is interesting recent work by Leong et al, who model the (scalar) parameter of interest using a dynamic Gauss-Markov process and assume the FC employs a Kalman filter to track the parameter \cite{Leong:2011,Leong:20112}. In \cite{Leong:2011}, both the orthogonal and coherent MAC were considered and two kinds of optimization problems were formulated: MSE minimization under a global sum transmit power constraint, and sum power minimization problem under an MSE constraint. An asymptotic expression for the MSE outage probability was also derived assuming a large number of sensor nodes.  The problem of minimizing the MSE outage probability for the orthogonal MAC with a sum power constraint was studied separately in \cite{Leong:20112}.

\subsection{Contributions}
In this paper, we consider scenarios similar to those in \cite{Leong:2011}.  In particular, we focus on the coherent MAC case assuming a dynamic parameter that is tracked via a Kalman filter at the FC.  As detailed in the list of contributions below, we extend the work of \cite{Leong:2011} for the case of a global sum power constraint, and we go beyond \cite{Leong:2011} to study problems where either the power of the individual sensors is constrained, or the goal is to minimize the peak power consumption of individual sensors:
\begin{enumerate}
\item We find a closed-form expression for the optimal complex transmission gains that minimize the MSE under a constraint on the sum power of all sensor transmissions.  While this problem was also solved in \cite{Leong:2011} using the KKT conditions derived in \cite{Xiao:2008}, our approach results in a simpler and more direct solution.  We also examine the asymptotic form of the solution for high total transmit power or high noise power at the FC.
\item We find a closed-form expression for the optimal complex transmission gain that minimizes the sum power under a constraint on the MSE.  In this case, the expression depends on the eigenvector of a particular matrix.  Again, while this problem was also addressed in \cite{Leong:2011}, the numerical solution therein is less direct than the one we obtain.  In addition, we find an asymptotic expression for the sum transmit power for a large number of sensors.
\item We show how to find the optimal transmission gains that minimize the MSE under individual sensor power constraints by relaxing the problem to a semi-definite programming (SDP) problem, and then proving that the optimal solution can be constructed from the SDP solution.
\item We show how to find the optimal transmission gains that minimize the maximum individual power over all of the sensors under a constraint on the maximum MSE.  Again, we solve the problem using SDP, and then prove that the optimal solution can be constructed from the SDP solution.
\item For the special case where the sensor nodes use equal power transmission, we derive an exact expression for the MSE outage probability.
%of MSE, and show that an increase in sum transmit power results in an
%exponential decrease in the outage probability.
\end{enumerate}
A subset of the above results were briefly presented in an earlier conference paper \cite{Jiang:2012}.

\subsection{Organization}
The rest of the paper is organized as follows. Section~\ref{sec:sys} describes the system model for the parameter tracking problem and provides an expression for the MSE obtained at the FC using a standard Kalman filter. Section~\ref{sec:msesum} investigates the MSE minimization problem under the assumption that the sensor nodes have a sum transmit power constraint. The MSE minimization problem with individual sensor power constraints is formulated and solved in Section~\ref{sec:mseindiv}.  The problems of minimizing the sum power or the maximum individual sensor power with MSE constraints are formulated and optimally solved in Section~\ref{powermin}.  In Section~\ref{outage}, the MSE outage probability for equal power allocation is derived.  Numerical results are presented in Section~\ref{sec:nume} and the conclusions are summarized in Section~\ref{sec:conc}.

% no \IEEEPARstart
%\cite{Joshi:1996,Leong:2010, Banavar:2010, Bajwa:2007, Cui:2007, Smith:2009, Xiao:2008}
\section{System Model}\label{sec:sys}
We model the evolution of a complex-valued dynamic parameter $\theta_n$ using a first-order Gauss-Markov process:
\begin{equation}
\theta_n=\alpha\theta_{n-1}+u_n\nonumber\;,
\end{equation}
where $n$ denotes the time step, $\alpha$ is the correlation parameter and the process noise $u_n$ is zero-mean complex normal with variance $\sigma_u^2$ (denoted by $\mathcal{CN}(0,\sigma_u^2)$).  We assume that $\theta_0$ is zero mean and that the norm $|\alpha| < 1$, so that $\theta_n$ is a stationary process.  Thus, the variance of $\theta_n$ is constant and given by $\sigma_\theta^2 = \sigma_u^2/\left(1-|\alpha|^2\right)$.  A set of $N$ sensors measures $\theta_n$ in the presence of noise; the measurement for the $i$th sensor at time $n$ is described by
\begin{equation}
s_{i,n} = \theta_n + v_{i,n} \; ,\nonumber
\end{equation}
where the measurement noise $v_{i,n}$ is distributed as $\mathcal{CN}(0,\sigma_{v,i}^2)$.  In an amplify-and-forward sensor network, each sensor multiplies its observation by a complex gain factor and transmits the result over a wireless channel to a fusion center (FC).  The FC receives a coherent sum of the signals from all $N$ sensors in additive noise:
\begin{eqnarray*}
y_n & = & \sum_{i=1}^N h_{i,n} a_{i,n} s_{i,n} + w_n \\
& = & \sum_{i=1}^N (h_{i,n} a_{i,n} \theta_n + h_{i,n} a_{i,n} v_{i,n}) + w_n \; ,
\end{eqnarray*}where $h_{i,n}$ is the gain of the wireless channel between sensor $i$ and the FC, $a_{i,n}$ is the complex transmission gain of sensor $i$, and $w_n$ is noise distributed as $\mathcal{CN}(0,\sigma_w^2)$.  This model can be written more compactly in matrix-vector form, as follows:
\begin{equation}
y_n=\mathbf{a}_n^H\mathbf{h}_n\theta_n+\mathbf{a}_n^H\mathbf{H}_n\mathbf{v}_n+w_n\;,\nonumber
\end{equation}
where $\mathbf{h}_n=[h_{1,n},\dots,h_{N,n}]^T$, $(\cdot)^T$ and $(\cdot)^H$ denote the transpose and complex conjugate transpose respectively, $\mathbf{a}_n=[a_{1,n},\dots,a_{N,n}]^{H}$ is a vector containing the conjugate of the sensor transmission gains, $\mathbf{H}_n=\mathrm{diag}\{h_{1,n},\dots,h_{N,n}\}$ is a diagonal matrix, and the measurement noise vector $\mathbf{v}_n=[v_{1,n},\dots,v_{N,n}]^T$ has covariance
$\mathbf{V}=\mathbb{E}\{\mathbf{v}_n\mathbf{v}_n^{H}\}=
\mathrm{diag}\left\{\sigma_{v,1}^2,\cdots,\sigma_{v,N}^2\right\}$.

The FC is assumed to know the statistics of the various noise processes, the current channel state $\mathbf{h}_n$, and the transmission gains $\mathbf{a}_n$, and it uses a standard Kalman filter to track the parameter $\theta_n$ according to the equations below \cite{Kay:1993}:
\begin{itemize}
\item Prediction Step: $\hat{\theta}_{n|n-1}=\alpha\hat{\theta}_{n-1|n-1}$
\item Prediction MSE:
$P_{n|n-1}=\alpha^2P_{n-1|n-1}+\sigma_u^2$
\item Kalman Gain:
\begin{equation}
k_n=\frac{P_{n|n-1}\mathbf{h}_n^H\mathbf{a}_n}{\mathbf{a}_n^H\mathbf{H}_n\mathbf{V}\mathbf{H}_n^H\mathbf{a}_n+P_{n|n-1}\mathbf{a}_n^H\mathbf{h}_n\mathbf{h}_n^H\mathbf{a}_n+\sigma_w^2}\nonumber
\end{equation}
\item Measurement Update:
\begin{equation}
\hat{\theta}_{n|n}=\hat{\theta}_{n|n-1}+k_n\left(y_n-\mathbf{a}_n^H\mathbf{h}_n\hat{\theta}_{n|n-1}\right)\nonumber
\end{equation}
\item Filtered MSE:
\begin{equation}\label{eq:mse}
P_{n|n}=(1-k_n\mathbf{a}_n^H\mathbf{h}_n)P_{n|n-1}\;.
\end{equation}
\end{itemize}
The goal is to determine an optimal choice for the gains $\mathbf{a}_n$ that minimizes the filtered MSE under a power constraint, or that minimizes the power consumed in transmitting the data to the FC under an MSE constraint.  The optimal gains are then fed back to the individual sensors to use at time $n$.

\section{Minimizing MSE under a Power Constraint}\label{sec:msesum}

\subsection{Global Sum Power Constraint}
In this section, we briefly consider the problem of minimizing the MSE under the assumption that the sensor nodes have a sum power constraint.  As mentioned earlier, this problem has already been studied in \cite{Leong:2011}, but the solution we provide here is simpler and more direct.  The optimization problem can be written as
\begin{eqnarray}\label{eq:opt1}
\min_{\mathbf{a_n}} && P_{n|n}\\
s. t. &&\mathbf{a}_n^H\mathbf{D}\mathbf{a}_n\le P_{T}\nonumber\;,
\end{eqnarray}
where $\mathbf{a}_n^H \mathbf{D} \mathbf{a}_n$ and $P_{T}$ respectively represent the actual and total available transmit power, with $\mathbf{D}=\textrm{diag}\{\sigma_{\theta}^2+\sigma_{v,1}^2,\cdots,
\sigma_{\theta}^2+\sigma_{v,N}^2 \}$.  From~(\ref{eq:mse}), minimizing the MSE $P_{n|n}$ is equivalent to maximizing
\begin{eqnarray}
&&k_n\mathbf{a}_n^H\mathbf{h}_n=\frac{P_{n|n-1}\mathbf{a}_n^H\mathbf{h}_n\mathbf{h}_n^H\mathbf{a}_n}{\mathbf{a}_n^H\mathbf{H}_n\mathbf{V}\mathbf{H}^H_n\mathbf{a}_n+P_{n|n-1}\mathbf{a}_n^H\mathbf{h}_n\mathbf{h}_n^H\mathbf{a}_n+\sigma_w^2}\;,\nonumber
\end{eqnarray}
and after a simple manipulation, the optimization problem in~(\ref{eq:opt1}) is equivalent to
\begin{eqnarray}\label{eq:opt2}
\max_{\mathbf{a}_n} &&\frac{\mathbf{a}_n^H\mathbf{h}_n\mathbf{h}_n^H\mathbf{a}_n}{\mathbf{a}_n^H\mathbf{H}_n\mathbf{V}\mathbf{H}^H_n\mathbf{a}_n+\sigma_w^2}\\
s. t. &&\mathbf{a}_n^H\mathbf{D}\mathbf{a}_n\le P_{T}\;.\nonumber
\end{eqnarray}

Denote the optimal solution to~(\ref{eq:opt2}) as
$\mathbf{a}^{*}_n$. It is easy to verify that the objective function of (\ref{eq:opt2}) is monotonically increasing in the norm of $\mathbf{a}_n$, which implies that at the optimal solution, the sum transmit
power constraint should be met with equality
$\mathbf{a}^{*H}_n\mathbf{D}\mathbf{a}^{*}_n=P_{T}$.  Thus~(\ref{eq:opt2}) becomes
%%\begin{eqnarray}\label{eq:opt3}
%%\max_{\mathbf{a}_n} %%&&\frac{\mathbf{a}_n^H\mathbf{h}_n\mathbf{h}_n^H\mathbf{a}_n}{\mathbf{a}_n^H(\mathbf{H}_n\mathbf{V}\mathbf{H}^H_n+\frac{\sigma_w^2}{P_{\max}}\mathbf{D})\mathbf{a}_n}\\
%%s. t. &&\mathbf{a}_n^H\mathbf{D}\mathbf{a}_n=P_{\max}\; . \nonumber
%%\end{eqnarray}
a Rayleigh quotient under a quadratic equality constraint.  Since the numerator involves a rank-one quadratic term, a simple closed-form solution is possible. If we define
$\mathbf{B}=\mathbf{H}_n\mathbf{V}\mathbf{H}^H_n+\frac{\sigma_w^2}{P_{T}}\mathbf{D}$,
the optimal solution is given by
\begin{equation}\label{eq:optimal}
\mathbf{a}_n^{*}=\sqrt{\frac{P_{T}}{\mathbf{h}_n^H\mathbf{B}^{-1}\mathbf{D}\mathbf{B}^{-1}
\mathbf{h}_n}}\mathbf{B}^{-1}\mathbf{h}_n \;.
\end{equation}
Note that the phase of each sensor transmission gain is the conjugate of the channel to the FC (recall that $\mathbf{a}_n$ contains the conjugate of these transmission gains).  In \cite{Leong:2011}, this property was assumed from the beginning in order to get an optimization problem with only real-valued variables; however, we see that this phase-matched solution results even without this assumption.

The maximum value of the objective function in~(\ref{eq:opt2}) can be expressed as
\begin{eqnarray}\label{eq:maxvalue}
\frac{\mathbf{a}_n^{*H}\mathbf{h}_n\mathbf{h}_n^H\mathbf{a}_n^{*}}{\mathbf{a}_n^{*H}
(\mathbf{H}_n\mathbf{V}\mathbf{H}^H_n+\frac{\sigma_w^2}{P_T}\mathbf{D})\mathbf{a}_n^{*}}&=&\mathbf{h}_{n}^{H}\mathbf{B}^{-1}\mathbf{h}_n \; .\nonumber
\end{eqnarray}
Given that
\begin{eqnarray}
\mathbf{h}_{n}^{H}\mathbf{B}^{-1}\mathbf{h}_n&\overset{(a)}{<}&\mathbf{h}_{n}^{H}
(\mathbf{H}\mathbf{V}\mathbf{H}^H)^{-1}\mathbf{h}_n\label{eq:7}\\
&=&\sum_{i=1}^N\frac{1}{\sigma_{v,i}^2}\;,\label{eq:72}
\end{eqnarray}
where $(a)$ follows from $\mathbf{B}^{-1}\prec(\mathbf{H}\mathbf{V}\mathbf{H}^H)^{-1}$, a lower bound on the MSE can be obtained by plugging~(\ref{eq:72}) into (\ref{eq:mse}):
\begin{eqnarray}\label{eq:lb}
P_{n|n}&>&\left(1-\frac{1}{1+\frac{1}{\left(\sum_{i=1}^{N}\frac{1}{\sigma_{v,i}^2}\right)P_{n|n-1}}}\right)P_{n|n-1}\nonumber\\
&=&\frac{P_{n|n-1}}{1+\left(\sum_{i=1}^{N}\frac{1}{\sigma_{v,i}^2}\right)P_{n|n-1}}\;.
\end{eqnarray}
Equation~(\ref{eq:7}) becomes an equality when $\sigma_w^2/P_T \rightarrow 0$ or when the signal-to-noise-ratio (SNR) at the FC is very high, and the resulting optimal sensor transmission gains become
\begin{equation}\label{eq:optimal2}
\mathbf{a}_{n}^{*}\!=\!\!\sqrt{\frac{P_T}{\sum_{i=1}^N\frac{1}{\sigma_{v,i}^4|h_{n,i}|^2
(\sigma_{\theta}^2+\sigma_{v,i}^2)}}}\left[\frac{1}{h_{1,n}\sigma_{v,1}^2} \;, \cdots \;,  \frac{1}{h_{N,n}\sigma_{v,N}^2}\right]^H \; .
\end{equation}In this case, sensors with small channel gains or low measurement noise are allocated more transmit power. On the other hand, for low SNR at the FC where $\sigma_w^2/P_T \rightarrow \infty$, we have $\mathbf{B} \approx \frac{\sigma_w^2}{P_T} \mathbf{D}$, and hence from~(\ref{eq:optimal}) the optimal gain vector is proportional to
\begin{equation}\label{eq:optimal2a}
\mathbf{a}_{n}^{*}\! \propto \!\! \left[\frac{h_{1,n}}{\sigma_\theta^2+\sigma_{v,1}^2} \;, \cdots \;,  \frac{h_{N,n}}{\sigma_\theta^2+\sigma_{v,N}^2} \right]^H \; .
\end{equation}
Interestingly, unlike the high SNR case, for low SNR the sensors with large channel gains are assigned higher power.  This observation will be highlighted later in the simulations of Section~\ref{sec:nume}.

\subsection{Individual Power Constraints}\label{sec:mseindiv}

In a distributed sensor network, it is more likely that the power of the individual sensors would be constrained, rather than the total sum power of the network.  As seen in the previous section, when the SNR at the FC is high (low), a weak (strong) channel for a given sensor can lead to a high transmission power that the sensor may not be able to support.  Thus, in this section we address the problem of minimizing the MSE under individual sensor power constraints, as follows:
\begin{eqnarray}\label{eq:opt4}
\min_{\mathbf{a}_n} &&P_{n|n}\\
s. t. && |a_{i,n}|^2(\sigma_{\theta}^2+\sigma_{v,i}^2)\le P_{T,i}\;,\quad i=1,\cdots, N\nonumber\;,
\end{eqnarray}
where $P_{T,i}$ is the maximum transmit power available at the $i$th sensor
node. Similar to~(\ref{eq:opt1}), problem (\ref{eq:opt4}) can be rewritten as
\begin{eqnarray}\label{eq:opt42}
\max_{\mathbf{a}_n} &&\frac{\mathbf{a}_n^H\mathbf{h}_n\mathbf{h}_n^H\mathbf{a}_n}{\mathbf{a}_n^H\mathbf{H}_n\mathbf{V}\mathbf{H}^H_n\mathbf{a}_n+\sigma_w^2}\\
s. t. && |a_{i,n}|^2(\sigma_{\theta}^2+\sigma_{v,i}^2)\le P_{T,i}\;,\quad i=1,\cdots, N\nonumber\;.
\end{eqnarray}
Problem (\ref{eq:opt42}) is a quadratically constrained ratio of
quadratic functions (QCRQ), and as explained below we will use the approach of \cite{Beck:2010} to transform the QCRQ problem into a relaxed SDP problem.

Introduce a real auxiliary variable $t$ and define
$\tilde{\mathbf{a}}_{n}=t\mathbf{a}_n$, so that problem (\ref{eq:opt42})
is equivalent to
\begin{eqnarray}\label{eq:opt5}
\max_{\mathbf{a}_n, t} &&\frac{\tilde{\mathbf{a}}_n^H\mathbf{h}_n\mathbf{h}_n^H\tilde{\mathbf{a}}_n}{\tilde{\mathbf{a}}_n^H\mathbf{H}_n\mathbf{V}\mathbf{H}^H_n\tilde{\mathbf{a}}_n+\sigma_w^2t^2}\\
s. t. &&\tilde{\mathbf{a}}_n^H\mathbf{D}_i\tilde{\mathbf{a}}_n\le t^2P_{T,i}\;,\quad i=1,\cdots, N\nonumber\\
      &&t\neq 0\;, \nonumber
\end{eqnarray}
where $\mathbf{D}_i=\textrm{diag}\{0,\cdots,0,
\sigma_{\theta}^2+\sigma_{v,i}^2,0,\cdots,0\}$.  We can further
rewrite problem (\ref{eq:opt5}) as
\begin{eqnarray}\label{eq:opt6}
\max_{\mathbf{a}_n, t} &&\tilde{\mathbf{a}}_n^H\mathbf{h}_n\mathbf{h}_n^H\tilde{\mathbf{a}}_n\\
s. t. &&\tilde{\mathbf{a}}_n^H\mathbf{H}_n\mathbf{V}\mathbf{H}^H_n\tilde{\mathbf{a}}_n+\sigma_w^2t^2=1 \nonumber\\
      &&\tilde{\mathbf{a}}_n^H\mathbf{D}_i\tilde{\mathbf{a}}_n\le t^2P_{T,i},\quad i=1,\cdots, N\;.\nonumber
\end{eqnarray}
Note that the constraints in problem (\ref{eq:opt6}) already guarantee
that $t\neq 0$, so this constraint is removed.

Define $\bar{\mathbf{a}}_n=[\tilde{\mathbf{a}}^{H}_n \; t]^{H}$ and the matrices
\begin{equation}
\bar{\mathbf{H}}_n=\left[\begin{array}{cc}\mathbf{h}_n\mathbf{h}_n^H&0\\
                \mathbf{0}^T& \mathbf{0}
                \end{array}\right],\qquad
                \bar{\mathbf{C}}_n=\left[\begin{array}{cc}\mathbf{H}_n\mathbf{V}\mathbf{H}^H_n& \mathbf{0}\\
                \mathbf{0}^T&\sigma_w^2
                \end{array}\right],\qquad
                \bar{\mathbf{D}}_i=\left[\begin{array}{cc}\mathbf{D}_i& \mathbf{0}\\
                \mathbf{0}^T&-P_{T,i}
               \end{array}\right] \; ,\nonumber
\end{equation}
so that problem (\ref{eq:opt6}) can be written in the compact form
\begin{eqnarray}\label{eq:opt7}
\max_{\bar{\mathbf{a}}_n} &&\bar{\mathbf{a}}_n^H\bar{\mathbf{H}}_n\bar{\mathbf{a}}_n\\
s. t. &&\bar{\mathbf{a}}_n^H\bar{\mathbf{C}}_n\bar{\mathbf{a}}_n=1 \nonumber\\
      &&\bar{\mathbf{a}}_n^H\bar{\mathbf{D}}_i\bar{\mathbf{a}}_n\le 0\;,\quad i=1,\cdots, N\;. \nonumber
\end{eqnarray}
Defining the $(N+1)\times(N+1)$ matrix $\bar{\mathbf{A}} = \bar{\mathbf{a}}_n \bar{\mathbf{a}}_n^H$, problem~(\ref{eq:opt7}) is equivalent to
\begin{eqnarray}\label{eq:opt8}
\max_{\bar{\mathbf{A}}}&&\mathrm{tr}(\bar{\mathbf{A}}\bar{\mathbf{H}}_n)\\
s. t. &&\mathrm{tr}(\bar{\mathbf{A}}\bar{\mathbf{C}}_n)=1 \nonumber\\
&&\mathrm{tr}(\bar{\mathbf{A}}\bar{\mathbf{D}}_i)\le 0\;,\quad i=1,\cdots, N\nonumber\\
&&\mathrm{rank}(\bar{\mathbf{A}})=1 \nonumber\\
&&\bar{\mathbf{A}}\succeq 0\;.\nonumber
\end{eqnarray}

Were it not for the rank constraint, the problem in~(\ref{eq:opt8}) would be a standard SDP problem and could be solved in polynomial time using (for example) the interior point method.  Given the difficulty of handling the rank constraint, we choose to relax it and solve the simpler problem
\begin{eqnarray}\label{eq:opt9}
\max_{\bar{\mathbf{A}}}&&\mathrm{tr}(\bar{\mathbf{A}}\bar{\mathbf{H}}_n)\\
s. t. &&\mathrm{tr}(\bar{\mathbf{A}}\bar{\mathbf{C}}_n)=1  \nonumber\\
&&\mathrm{tr}(\bar{\mathbf{A}}\bar{\mathbf{D}}_i)\le 0\;,\quad i=1,\cdots, N\;.\nonumber\\
&&\bar{\mathbf{A}}\succeq 0\;,\nonumber
\end{eqnarray}
which would provide an upper bound on the optimal value of problem (\ref{eq:opt42}), and would in general lead to a suboptimal solution for the vector $\mathbf{a}_n$ of transmission gains.  However, in the following we show that the optimal solution to the original problem in~(\ref{eq:opt4}) can be constructed from the solution to the relaxed SDP problem in~(\ref{eq:opt9}).  The optimality of a rank-relaxed SDP problem similar to the one we consider here has previously been noted in \cite{Liao:2011}, but for a different problem related to physical layer security.  To describe how to find the optimal solution from the rank-relaxed problem in~(\ref{eq:opt9}), define $\bar{\mathbf{A}}^{*}$ to be the solution to~(\ref{eq:opt9}), $\bar{\mathbf{A}}^{*}_{l,m}$ as the $(l,m)$th element of
$\bar{\mathbf{A}}^{*}$, and $\bar{\mathbf{A}}_N^{*}$ as the $N$th order leading principal submatrix of $\bar{\mathbf{A}}^{*}$ formed by deleting the $(N+1)$st row and column of $\bar{\mathbf{A}}^{*}$.  Then the optimal solution can be found via the following theorem. 

\begin{theorem}\label{theorem1}
Define the optimal solution to problem (\ref{eq:opt9}) as
$\bar{\mathbf{A}}^{*}$. Then
$\bar{\mathbf{A}}^{*}_N=\mathbf{a}\mathbf{a}^{H}$ is rank-one and the optimal
solution to problem (\ref{eq:opt4}) is given by
\begin{equation}
\mathbf{a}^{*}_n=\frac{1}{\sqrt{\bar{\mathbf{A}}_{N+1,N+1}^{*}}}\mathbf{a}\;.\nonumber
\end{equation}
\end{theorem}
\begin{IEEEproof}
We first utilize the strong duality between problem (\ref{eq:opt9})
and its dual to find properties of the optimal solution
$\bar{\mathbf{A}}^{*}$. The dual of problem (\ref{eq:opt9}) is
given by \cite{Zhang:2011}:
\begin{eqnarray}\label{eq:opt10}
\min_{y_i,z}&& z\\
s. t. &&\sum_{i=1}^{N}y_i\bar{\mathbf{D}}_i+z\bar{\mathbf{C}}_n-\bar{\mathbf{H}}_n\succeq 0 \nonumber\\
&&y_1,\dots,y_N,z\ge 0\nonumber\;.
\end{eqnarray}
It is easy to verify that there exist strictly feasible points for
problems~(\ref{eq:opt9}) and~(\ref{eq:opt10}).  In particular, for~(\ref{eq:opt9}), we can construct
\begin{equation}
\bar{\mathbf{A}}^f=\textrm{diag}\{ab, \cdots, ab, b\}\;,\nonumber
\end{equation}
where
\begin{align*}
0 < &a <\min_{i}\frac{P_{T,i}}{\sigma_{\theta}^2+\sigma_{v,i}^2}\;, \\
&b = \frac{1}{\sum_{i=1}^{N}a|h_{n,i}|^2\sigma_{v,i}^2+\sigma_w^2} \; .
\end{align*}
For~(\ref{eq:opt10}), we can randomly select $y_i^f>0$, and set
$z^f$ large enough such that
\begin{equation} z^f\!>\max\left\{\!\frac{\mathbf{h}_{n}^{H}\mathbf{h}_n\!+\!\!\sum_{i=1}^Ny_i^fP_{T,i}}{\sigma_w^2},\frac{\mathbf{h}_n^{H}\mathbf{h}_n\!-\!y_i^f(\sigma_{\theta}^2\!+\!\sigma_{v,i}^2)}{|h_{n,i}|^2\sigma_{v,i}^2}\!\right\}\;.\nonumber
\end{equation}
Then, according to Slater's theorem, strong duality holds between the primal problem (\ref{eq:opt9}) and the dual problem (\ref{eq:opt10}) and we have the following complementary condition:
\begin{equation}\label{eq:slater}
\textrm{tr}(\bar{\mathbf{A}}^{*}\mathbf{G}^{*})=0\;,
\end{equation}
where $\mathbf{G}^{*}=\sum_{i=1}^{N}y_i^{*}\bar{\mathbf{D}}_i+z^{*}\bar{\mathbf{C}}_n-\bar{\mathbf{H}}_n$ and $y_{i}^{*}$ and $z^{*}$ denote the optimal solution to problem (\ref{eq:opt10}).
Due to the special structure of $\bar{\mathbf{D}}_i$, $\bar{\mathbf{C}}_n$ and $\bar{\mathbf{H}}_n$, $\mathbf{G}^{*}$ can be expressed as
\begin{equation}
\mathbf{G}^{*}=\left[\begin{array}{cc}\mathbf{G}^{*}_N& \mathbf{0}\\
                \mathbf{0}^T&\mathbf{G}^{*}_{N+1,N+1}
                \end{array}\right]\;,\nonumber
\end{equation}
where $\mathbf{G}^{*}_N=\sum_{i=1}^Ny_i^*\mathbf{D}_i+z^*\mathbf{H}_n\mathbf{V}\mathbf{H}_n^H-
\mathbf{h}_n\mathbf{h}_n^H$ and $\mathbf{G}^{*}_{N+1,N+1}=z^{*}\sigma_w^2-\sum_{i=1}^{N}y_i^{*}P_{T,i}$.
Since both $\bar{\mathbf{A}}^{*}$ and $\mathbf{G}^{*}$ are positive semidefinite, (\ref{eq:slater}) is equivalent to
\begin{equation}
\bar{\mathbf{A}}^{*}\mathbf{G}^{*}=0\;.\nonumber
\end{equation}
\iffalse
\begin{equation}
\mathbf{G}^{*}_{N+1,N+1}=0
\end{equation}
\fi
Additionally, with consideration of the structure of $\mathbf{G}^{*}$, we have
\begin{equation}
\bar{\mathbf{A}}^{*}_N\mathbf{G}^{*}_N=0\;.\nonumber
\end{equation}

Define $\mathbf{V}_{G}$ as a set of vectors
orthogonal to the row space of
$\mathbf{G}^{*}_{N}$.  Then the column vectors of
$\bar{\mathbf{A}}_{N}^{*}$ must belong to span($\mathbf{V}_G$) and
$\textrm{rank}(\bar{\mathbf{A}}^{*}_N)\le
\textrm{rank}(\mathbf{V}_G)$.  For any two matrices $\mathbf{M}$ and
$\mathbf{N}$, we have \cite{Gentle:2007} that
$\textrm{rank}(\mathbf{M}+\mathbf{N})\ge|\textrm{rank}(\mathbf{M})-\textrm{rank}(\mathbf{N})|$, so
\begin{eqnarray}
\textrm{rank}(\mathbf{G}_N^{*})\!\!\!&\ge&\!\!\!\textrm{rank}\left(\sum_{i=1}^Ny_i^*\mathbf{D}_i+z^*\mathbf{H}_n\mathbf{V}\mathbf{H}_n^H\right)\!\!-\textrm{rank}(\mathbf{h}_n\mathbf{h}_n^H)\nonumber\\
&=&\!\!\!N-1\;.\nonumber
\end{eqnarray}
and
\begin{equation}\label{eq:rank1}
\textrm{rank}(\mathbf{V}_G) = N-\textrm{rank}(\mathbf{G}_N^{*}) \le 1 \;.
\end{equation}
Since $\text{tr}(\bar{\mathbf{A}}^{*}\bar{\mathbf{H}})=\mathbf{h}_n^H\bar{\mathbf{A}}^*_N\mathbf{h}_n$ and $\text{tr}(\bar{\mathbf{A}}^{*}\bar{\mathbf{H}})>\text{tr}(\bar{\mathbf{A}}^{f}\bar{\mathbf{H}})>0$, we have
\begin{equation}\label{eq:rank2}
\bar{\mathbf{A}}^*_N\neq 0\;, \qquad \textrm{rank}(\bar{\mathbf{A}}^*_N)\ge 1\;.
\end{equation}
Combining (\ref{eq:rank1}) and (\ref{eq:rank2}) then leads to
\begin{equation}
\textrm{rank}(\bar{\mathbf{A}}^*_N)=1\;.\nonumber
\end{equation}

Although at this point we don't know whether the optimal solution $\bar{\mathbf{A}}^*$ is rank-one, we can construct a rank-one optimal solution based on $\bar{\mathbf{A}}^*$.
Define the rank-one decomposition of $\bar{\mathbf{A}}_N^{*}$ as
$\bar{\mathbf{A}}_N^{*}=\mathbf{a}\mathbf{a}^{H}$, so that the optimal
rank-one solution to problem (\ref{eq:opt9}) is
\begin{equation}\label{eq:rone}
\bar{\mathbf{A}}^{'}=\bar{\mathbf{a}}^{*}\bar{\mathbf{a}}^{*H},
\end{equation}
where $\bar{\mathbf{a}}^{*}=\left[\mathbf{a}^{H} \; \sqrt{\bar{\mathbf{A}}_{N+1,N+1}^*}\right]^{H}$. It is easy to verify that the rank-one matrix $\bar{\mathbf{A}}^{'}$ can achieve the same result for problem (\ref{eq:opt9}) as  $\bar{\mathbf{A}}^*$.

Since (\ref{eq:opt42}) is equivalent to problem (\ref{eq:opt4}) and (\ref{eq:opt8}), and (\ref{eq:opt9}) is realized from problem (\ref{eq:opt8}) by relaxing the rank-one constraint, in general the solution to~(\ref{eq:opt9}) provides an upper bound on the optimal value achieved by~(\ref{eq:opt42}). If the optimal solution to~(\ref{eq:opt4}) is $\mathbf{a}^{*}_n$, then
\begin{equation}\label{eq:ub}
\frac{\mathbf{a}_n^{*H}\mathbf{h}_n\mathbf{h}_n^H\mathbf{a}_n^{*}}{\mathbf{a}_n^{*H}\mathbf{H}\mathbf{V}\mathbf{H}^H\mathbf{a}_n^{*}+\sigma_w^2}\le\textrm{tr}(\bar{\mathbf{A}}^{*}\bar{\mathbf{H}})\;,
\end{equation}
where $\mathbf{a}_n^{*}$ and $\bar{\mathbf{A}}^{*}$ are the optimal solutions to problems (\ref{eq:opt4}) and (\ref{eq:opt9}) respectively. Equality can be achieved in ~(\ref{eq:ub}) provided that an optimal rank-one solution exists for~(\ref{eq:opt9}), and (\ref{eq:rone}) indicates that such a rank-one solution exists. In the following, we will show how to construct $\mathbf{a}_n^{*}$ based on $\bar{\mathbf{A}}^{*}$. According to problem (\ref{eq:opt9}), since
$\textrm{tr}(\bar{\mathbf{A}}^{*}\bar{\mathbf{C}}_n)=1$ and $\bar{\mathbf{A}}\succeq 0$, then we have $\bar{\mathbf{A}}^{*}\neq 0$ and further $\bar{\mathbf{A}}^{*}_{N+1,N+1}>0$. Based on
$\bar{\mathbf{a}}^{*}$, the optimal solution to~(\ref{eq:opt4}) is given by
\begin{equation}\label{eq:opts}
\mathbf{a}^{*}_n=\frac{\bar{\mathbf{a}}^{*}}{\sqrt{\bar{\mathbf{A}}_{N+1,N+1}^*}} \;,
\end{equation}
and plugging (\ref{eq:opts}) into (\ref{eq:ub}) we have
\begin{equation}
\frac{\mathbf{a}_n^{*H}\mathbf{h}_n\mathbf{h}_n^H\mathbf{a}_n^{*}}{\mathbf{a}_n^{*H}\mathbf{H}\mathbf{V}\mathbf{H}^H\mathbf{a}_n^{*}+\sigma_w^2}=\textrm{tr}(\bar{\mathbf{A}}^{*}\bar{\mathbf{H}})\;,\nonumber
\end{equation}
which verifies the optimality of $\mathbf{a}^{*}_n$.
\end{IEEEproof}

\iffalse
In addition to the optimal numerical solution provided in Theorem \ref{theorem1}, we attempt to calculate approximate closed-form solution. When $P_{\max,i}\gg \sigma_w^2$, at the FC, the additive noise's effect can be neglected and the value of the objective function only depends on the direction of the power control parameter $\mathbf{a}$, and the optimal direction of $\mathbf{a}$ is given by Eq. (\ref{eq:lb}). To satisfy the individual power constraint, we need to scale $\mathbf{a}^{*}$ properly, and the approximate solution is
\begin{equation}
\mathbf{a}_{n}^{*}\!=\!\!p\left[\frac{1}{\bar{h}_{1,n}\sigma_{v,1}^2},\cdots, \frac{1}{\bar{h}_{N,n}\sigma_{v,N}^2}\right],
\end{equation}
where $p=\min_i|h_{i,n}|\sigma_{v,i}^2\sqrt{\frac{P_{\max,i}}{\sigma_{v,i}^2+\sigma_{\theta}^2}}$.\fi

%recover the rank one solution, if $\mathbf{a}^{*}^{H}\bar{\mathbf{D}}_i\mathbf{a}^{*}>0$, then set $a_{i}^{*}=$

\section{Minimizing Transmit Power under an MSE Constraint}\label{powermin}

In this section, we consider the converse of the problems investigated in Section~\ref{sec:msesum}.  We first look at the problem addressed in \cite{Leong:2011}, where the goal is to minimize the sum power consumption of all the sensors under the constraint that the MSE is smaller than some threshold.  The asymptotic behavior of the solution is then characterized for a large number of sensors, $N$.  Next we study the case where the maximum individual transmit power of any given sensor is minimized under the MSE constraint.

\subsection{Minimizing Sum Transmit Power}

We can express the problem of minimizing the sum transmit power under the constraint that the MSE is smaller than $\epsilon$ as follows:
\iffalse
{\bf discussion of the optimal solution with the minimize mse case and similar results were obtained in \cite{Cui:2007,Leong:2010}}\fi
\begin{eqnarray}\label{eq:minsp}
\min_{\mathbf{a}_n} &&\mathbf{a}_n^H\mathbf{D}\mathbf{a}_n\\
s. t. &&P_{n|n}\le \epsilon \; . \nonumber
\end{eqnarray}
To make~(\ref{eq:minsp}) feasible, according to~(\ref{eq:mse}) and~(\ref{eq:lb}) the value of $\epsilon$ should satisfy
\begin{equation}\label{eq:feas}
\frac{P_{n|n-1}}{1+\left(\sum_{i=1}^{N}\frac{1}{\sigma_{v,i}^2}\right)P_{n|n-1}} \le \epsilon \le P_{n|n-1}\;.
\end{equation}
As discussed earlier, the MSE is monotonically decreasing in the norm of $\mathbf{a}_n$, so it is clear that setting $P_{n|n}=\epsilon$ results in the minimum possible transmit power, which we refer to as $P_T^*$.  Conceptually, the problem can be solved by finding the value of $P_T^*$ for which $P_{n|n}=\epsilon$, and then substituting this value into the solution found in~(\ref{eq:optimal}):
\begin{equation}
\mathbf{a}_n^{*}=\sqrt{\frac{P_T^*}{\mathbf{h}_n^H\mathbf{B}^{-1}
\mathbf{D}\mathbf{B}^{-1}\mathbf{h}_n}}\mathbf{B}^{-1}\mathbf{h}_n \; .\nonumber
\end{equation}
Unlike \cite{Leong:2011}, where an unspecified numerical procedure was required to solve this problem, in the following we present a direct ``closed-form'' solution that finds the result in terms of the eigenvalue and eigenvector of a particular matrix.

Assuming that $\epsilon$ satisfies the feasibility constraint of~(\ref{eq:feas}), we use~(\ref{eq:mse}) and $P_{n|n}=\epsilon$ to convert~(\ref{eq:minsp}) to the following form:
\iffalse \begin{eqnarray}
\min_{\mathbf{a}_n} &&\mathbf{a}_n^H\mathbf{D}\mathbf{a}_n\nonumber \\
s. t. &&\frac{P_{n|n-1}\mathbf{a}_n^H\mathbf{h}_n\mathbf{h}_n^H\mathbf{a}_n}{\mathbf{a}_n^H\mathbf{H}\mathbf{V}
\mathbf{H}^H\mathbf{a}_n+\sigma_w^2} = \frac{P_{n|n-1}}{\epsilon}-1
\end{eqnarray}\fi
\begin{eqnarray}\label{eq:minspequ}
\min_{\mathbf{a}_n} &&\mathbf{a}_n^H\mathbf{D}\mathbf{a}_n\\
s. t. &&\mathbf{a}^{H}_{n}\mathbf{E}_n\mathbf{a}_n\ge\left(\frac{P_{n|n-1}}{\epsilon}-1\right)\sigma_{w}^2\;,\nonumber
\end{eqnarray}
where $\mathbf{E}_n=P_{n|n-1}\mathbf{h}_n\mathbf{h}_n^H-\left(\frac{P_{n|n-1}}{\epsilon}-1\right)
\mathbf{H}_n\mathbf{V}\mathbf{H}^H_n$. It's obvious that the constraint in problem (\ref{eq:minspequ}) should be active at the optimal solution and we can rewrite problem (\ref{eq:minspequ}) as 
\begin{eqnarray}\label{eq:minspequ3}
\min_{\mathbf{a}_n} &&\frac{\mathbf{a}_n^H\mathbf{D}\mathbf{a}_n}{\mathbf{a}^{H}_{n}\mathbf{E}_n\mathbf{a}_n}\\
s. t. &&\mathbf{a}^{H}_{n}\mathbf{E}_n\mathbf{a}_n=\left(\frac{P_{n|n-1}}{\epsilon}-1\right)\sigma_{w}^2\;.\nonumber
\end{eqnarray}
Since both of $\mathbf{a}_n^H\mathbf{D}\mathbf{a}_n$ and $\mathbf{a}^{H}_{n}\mathbf{E}_n\mathbf{a}_n$ are positive, problem (\ref{eq:minspequ3}) is equivalent to
\begin{eqnarray}\label{eq:minspequ4}
\max_{\mathbf{a}_n} &&\frac{\mathbf{a}^{H}_{n}\mathbf{E}_n\mathbf{a}_n}{\mathbf{a}_n^H\mathbf{D}\mathbf{a}_n}\\
s. t. &&\mathbf{a}^{H}_{n}\mathbf{E}_n\mathbf{a}_n=\left(\frac{P_{n|n-1}}{\epsilon}-1\right)
\sigma_{w}^2\nonumber \; .
\end{eqnarray}
Setting $\mathbf{y}=\mathbf{D}^{\frac{1}{2}}\mathbf{a}_n$, problem (\ref{eq:minspequ4}) becomes a Rayleigh quotient maximization:
\begin{eqnarray}
\max_{\mathbf{y}} &&\frac{\mathbf{y}^{H}\mathbf{D}^{-\frac{1}{2}}\mathbf{E}_n\mathbf{D}^{-\frac{1}{2}}\mathbf{y}}{\mathbf{y}^H\mathbf{y}}\nonumber\\
s. t. &&\mathbf{y}^{H}\mathbf{D}^{-\frac{1}{2}}\mathbf{E}_n\mathbf{D}^{-\frac{1}{2}}\mathbf{y}=
\left(\frac{P_{n|n-1}}{\epsilon}-1\right)\sigma_{w}^2\nonumber \; ,
\end{eqnarray}
whose solution is given by
\begin{equation}
\mathbf{y}^{*}=\sqrt{\frac{\left(\frac{P_{n|n-1}}{\epsilon}-1\right)
\sigma_w^2}{\mathbf{v}_1\mathbf{D}^{-\frac{1}{2}}\mathbf{E}_n\mathbf{D}^{-\frac{1}{2}}
\mathbf{v}_1}}\mathbf{v}_1 \; \nonumber,
\end{equation}
where $\mathbf{v}_1$ denotes the unit-norm eigenvector corresponding to the largest eigenvalue of $\mathbf{D}^{-\frac{1}{2}}\mathbf{E}_n\mathbf{D}^{-\frac{1}{2}}$.  The optimal solution to the original problem in~(\ref{eq:minsp}) is thus
\begin{equation}
\mathbf{a}_n^*=\sqrt{\frac{\left(\frac{P_{n|n-1}}{\epsilon}-1\right)\sigma_w^2}
{\mathbf{v}_1\mathbf{D}^{-\frac{1}{2}}\mathbf{E}_n\mathbf{D}^{-\frac{1}{2}}\mathbf{v}_1}}
\mathbf{D}^{-\frac{1}{2}}\mathbf{v}_1 \;\nonumber .
\end{equation}

The minimum transmit power required to achieve $P_{n|n}=\epsilon$ can be expressed as
\begin{equation}\label{eq:pmax}
P_T^{*}= \mathbf{a}_n^{*H} \mathbf{D} \mathbf{a}_n^* =
\frac{(P_{n|n-1}-\epsilon)\sigma_w^2}{\epsilon\lambda_{\max}\{\mathbf{D}^{-\frac{1}{2}}
\mathbf{E}_n\mathbf{D}^{-\frac{1}{2}}\}} \; ,
\end{equation}
where $\lambda_{\max}(\cdot)$ represents the largest eigenvalue of its matrix argument.  A more precise expression for $P_T^*$ can be found when the number of sensors $N$ is large, as shown in Theorem~\ref{theorem2} below.  The theorem assumes that the channel coefficients are described by the following model:
\begin{equation}\label{eq:hmodel}
h_{i,n}=\frac{\tilde{h}_{i,n}}{d_{i}^\gamma},\quad \tilde{h}_{i,n} \sim \mathcal{CN}(0,1) \; ,
\end{equation}
where $d_i$ denotes the distance between sensor $i$ and the FC, and $\gamma$ is the propagation path-loss exponent.
\begin{theorem}\label{theorem2}
Assume the channels between the sensors and FC obey the model of~(\ref{eq:hmodel}).  When the number of sensors is large, the minimum sum transmit power $P_T^*$ that achieves $P_{n|n}=\epsilon$ is bounded by
\begin{equation*}
\frac{\left(P_{n|n-1}-\epsilon\right)\sigma_w^2}{\epsilon (P_{n|n-1}\mathbf{h}_n^H\mathbf{D}^{-1}\mathbf{h}_n-\xi)}<P_T^{*}<\frac{\left(P_{n|n-1}-
\epsilon\right)\sigma_w^2}{\epsilon P_{n|n-1}\mathbf{h}_n^H\mathbf{D}^{-1}\mathbf{h}_n(1-\zeta)} \; ,
\end{equation*} where random variables $\zeta$, $\xi$ are defined as 
\begin{eqnarray*}
\xi & = & \left(\frac{P_{n|n-1}}{\epsilon}-1\right)\min_{i}\left\{\frac{|h_{i,n}|^2\sigma_{v,i}^2}
{\sigma_{\theta}^2+\sigma_{v,i}^2}\right\} \\
\zeta & = & \frac{\left(\frac{P_{n|n-1}}{\epsilon}-1\right)\max_{i}\left\{\frac{|h_{i,n}|^2\sigma_{v,i}^2}
{\sigma_{\theta}^2+\sigma_{v,i}^2}\right\}}{P_{n|n-1}\mathbf{h}_n^H\mathbf{D}^{-1}\mathbf{h}_n}\;,
\end{eqnarray*}
and $\zeta$, $\xi$ converge to $0$ in probability.
\end{theorem}
\begin{proof}
See Appendix A.
\end{proof}

According to the above theorem, when $N\to\infty$, the term $P_{n|n-1}\mathbf{h}_n^H\mathbf{D}^{-1}\mathbf{h}_n$ is the dominant factor in the denominator of the bounds on the sum transmit power, and we have the following asymptotic expression
\begin{equation}\label{eq:approxsolu}
\lim_{N\to\infty}P_T^{*}\simeq\frac{(P_{n|n-1}-\epsilon)\sigma_w^2}{\epsilon P_{n|n-1}\mathbf{h}_n^{H}\mathbf{D}^{-1}\mathbf{h}_n}\;.
\end{equation}
This expression illustrates that to achieve the same MSE, increasing the number of sensors reduces the total required transmit power of the network, as well as the required transmit power per sensor.  A similar observation was made in \cite{Leong:2011}. As shown later, our simulation results show that~(\ref{eq:approxsolu}) provides an accurate approximation to (\ref{eq:pmax}) as long as $\epsilon$ is not too small.

As a final comment on this problem, we note that~(\ref{eq:minsp}) is equivalent to
\begin{eqnarray}\label{eq:minspequ2}
\min_{\mathbf{A}} &&\textrm{tr}(\mathbf{A}\mathbf{D})\\
s. t. &&\textrm{tr}(\mathbf{A}\mathbf{E}_n)\ge\left(\frac{P_{n|n-1}}{\epsilon}-1\right)\sigma_{w}^2 \nonumber\\
&&\mathrm{rank}(\mathbf{A})=1 \nonumber\\
&&\mathbf{A}\succeq 0 \nonumber
\end{eqnarray}
for $\mathbf{A}=\mathbf{a}_n \mathbf{a}_n^H$.  Relaxing the rank-one constraint on $\mathbf{A}$, problem (\ref{eq:minspequ2}) becomes
\begin{eqnarray}\label{eq:minspsdp}
\min_{\mathbf{A}} &&\textrm{tr}(\mathbf{A}\mathbf{D})\\
s. t. &&\textrm{tr}(\mathbf{A}\mathbf{E}_n)\ge\left(\frac{P_{n|n-1}}{\epsilon}-1\right)\sigma_{w}^2\;,\nonumber\\
&&\mathbf{A}\succeq 0\; . \nonumber
\end{eqnarray}
Based on the complementary conditions between the dual and primal problems, we can prove that the solution to~(\ref{eq:minspsdp}) is rank one, and hence that the relaxed SDP yields the optimal $\mathbf{a}_n^*$.

\subsection{Minimizing Maximum Individual Transmit Power}

Here we focus on the problem of minimizing the maximum transmit power of the individual sensors while attempting to meet an MSE objective:
\begin{eqnarray}\label{eq:minmax}
\min_{\mathbf{a}_n} \max_i &&|a_{i,n}|^2(\sigma_{\theta}^2+\sigma_{v,i}^2)\\
s. t. &&P_{n|n}\le \epsilon \; . \nonumber
\end{eqnarray}
As in Section~\ref{sec:mseindiv}, we will convert the problem to a rank-relaxed SDP whose solution nonetheless obeys the rank constraint and hence provides the optimal result.  To proceed,
introduce an auxiliary variable $t$ and rewrite~(\ref{eq:minmax}) as
\begin{eqnarray}\label{eq:equival}
\min_{\mathbf{a}_n, t} &&t\\
s. t. &&P_{n|n}\le \epsilon\nonumber\\
&&|a_{i,n}|^2(\sigma_{\theta}^2+\sigma_{v,i}^2)\le t,\quad i=1,\cdots, N\nonumber\;.
\end{eqnarray}
Problem (\ref{eq:equival}) is equivalent to
\begin{eqnarray}\label{eq:equival2}
\min_{\mathbf{A}, t} &&t\\
s. t. &&\mathrm{tr}(\mathbf{A}\mathbf{E}_n)-\left(\frac{P_{n|n-1}}{\epsilon}-1\right)\sigma_{w}^2\ge 0\nonumber\\
&&\mathrm{tr}(\mathbf{A}\mathbf{D}_i)-t\le 0,\; i=1,\cdots, N\nonumber\\
&&\mathbf{A} \succeq 0 \nonumber \\
&&\textrm{rank}\left(\mathbf{A}\right)=1\nonumber\;,
\end{eqnarray}
where $\mathbf{A}=\mathbf{a}_n\mathbf{a}_n^H$, $\mathbf{E}_n$ is defined as in~(\ref{eq:minspequ}), and
$\mathbf{D}_i=\textrm{diag}\{0,\cdots,\sigma_{\theta}^2+\sigma_{v,i}^2,0,\cdots,0\}$, as before.

Relaxing the rank constraint and rewriting the problem to be in standard form, problem~(\ref{eq:equival2}) becomes
\begin{eqnarray}\label{eq:relaxed}
\min_{\tilde{\mathbf{A}}} &&\textrm{tr}(\tilde{\mathbf{A}}\mathbf{T})\\
s. t. &&\mathrm{tr}(\tilde{\mathbf{A}}\tilde{\mathbf{E}}_n)-\left(\frac{P_{n|n-1}}{\epsilon}-1\right)
\sigma_{w}^2\ge 0\nonumber\\
&&\mathrm{tr}(\tilde{\mathbf{A}}\mathbf{F}_i)\le 0, \quad i=1,\cdots, N\nonumber\\
&&\tilde{\mathbf{A}}\succeq 0\nonumber\;,
\end{eqnarray}
where
\begin{equation*}
\tilde{\mathbf{A}}=\left[\begin{array}{cc}\mathbf{A}& \mathbf{w} \\
                \mathbf{w}^H & t
                \end{array}\right], \qquad
\mathbf{T}=\left[\begin{array}{cc}\mathbf{0}&\mathbf{0}\\
                \mathbf{0}&1
                \end{array}\right], \qquad
\tilde{\mathbf{E}}_n=\left[\begin{array}{cc}\mathbf{E}_n& \mathbf{0} \\
                \mathbf{0}^T&0
                \end{array}\right], \qquad
\mathbf{F}_i=\left[ \begin{array}{cc} \mathbf{D}_i & \mathbf{0} \\ \mathbf{0}^T & -1 \end{array}
                \right] \; ,
\end{equation*}
and $\mathbf{w}$ is otherwise arbitrary.
Theorem~\ref{theorem3} establishes that the optimal solution to~(\ref{eq:minmax}) can be constructed from the solution to the above relaxed SDP.
\iffalse
The dual problem is given by \cite{Huang:2010}
\begin{eqnarray}
\min_{\tilde{\mathbf{A}}} &&\textrm{tr}(\tilde{\mathbf{A}}\tilde{\mathbf{D}})\\
s. t. &&\mathrm{tr}(\tilde{\mathbf{A}}\tilde{\mathbf{E}})-\left(\frac{P_{n|n-1}}{\epsilon}-1\right)\sigma_{w}^2\ge 0\nonumber\\
&&\mathrm{tr}(\tilde{\mathbf{A}}\mathbf{F}_i)\le 0, \nonumber\\
&&\tilde{\mathbf{A}}\succeq 0\nonumber.
\end{eqnarray}\fi
\begin{theorem}\label{theorem3}
Define the optimal solution to problem~(\ref{eq:relaxed}) as
$\tilde{\mathbf{A}}^{*}$. Then
$\tilde{\mathbf{A}}^{*}_N=\tilde{\mathbf{a}}\tilde{\mathbf{a}}^{H}$ is rank-one and the optimal
solution to problem~(\ref{eq:minmax}) is given by
$\mathbf{a}^{*}_n=\tilde{\mathbf{a}}$\;.
\end{theorem}
\begin{proof}
The dual of problem (\ref{eq:relaxed}) is given by
\begin{eqnarray}\label{eq:dual}
\max_{y_i,z} && \left(\frac{P_{n|n-1}}{\epsilon}-1\right)\sigma_{w}^2z\\
s. t. &&\mathbf{T}+\sum_{i=1}^{N}y_i\mathbf{F}_i-z\tilde{\mathbf{E}}_n\succeq 0\nonumber\\
&&y_1,\cdots,y_N,z\ge 0\;.\nonumber
\end{eqnarray}
Using an approach similar to the proof of Theorem~\ref{theorem1}, one can verify that both~(\ref{eq:relaxed}) and (\ref{eq:dual}) are strictly feasible, and that strong duality holds between the dual problem~(\ref{eq:dual}) and the primal problem~(\ref{eq:relaxed}). Based on the complementary conditions, it can be shown that $\text{rank}(\tilde{\mathbf{A}}_{N}^{*})=1$. For brevity the details of the proof are omitted.
\end{proof}

Similar to problems~(\ref{eq:opt1}) and (\ref{eq:minsp}), duality also exists between~(\ref{eq:opt4}) and~(\ref{eq:minmax}). 	Define the optimal solution to problem~(\ref{eq:opt4}) as $\mathbf{a}_n^{*}$ and the corresponding minimum MSE as $P_{n|n}^{*}$.  If we set $\epsilon=P_{n|n}^{*}$ in~(\ref{eq:minmax}), the optimal solution is also $\mathbf{a}_n^{*}$.

\section{MSE Outage Probability for Equal Power Allocation}\label{outage}

Here we calculate the MSE outage probability for the suboptimal
solution in which each sensor transmits with the same power.
The outage probability derived here can serve as an upper
bound for the outage performance of the optimal algorithm
with individual power constraints.  For equal-power transmission,
the transmit gain vector is given by
\begin{equation}
\mathbf{a}_{e}=\sqrt{\frac{P_T}{N}}\left[\frac{1}{\sqrt{\sigma_{\theta}^2+\sigma_{v,1}^2}},\cdots, \frac{1}{\sqrt{(\sigma_{\theta}^2+\sigma_{v,N}^2)}}\right]^T\;\nonumber,
\end{equation}
and the corresponding MSE is
\begin{equation}
P_{n|n}=\left(1-\frac{P_{n|n-1}\mathbf{a}_e^H\mathbf{h}_n\mathbf{h}_n^H\mathbf{a}_e}
{\mathbf{a}_e^H\mathbf{H}_n\mathbf{V}\mathbf{H}_n\mathbf{a}_e+P_{n|n-1}\mathbf{a}_e^H\mathbf{h}_n
\mathbf{h}_n^H\mathbf{a}_e+\sigma_w^2}\right)P_{n|n-1} \;\nonumber .
\end{equation}
As in Theorem~\ref{theorem2}, we will assume the Gaussian channel model of~(\ref{eq:hmodel}).
The outage probability $P_{out}=\mathrm{Pr}\left\{P_{n|n}>\epsilon\right\}$ is evaluated as follows:
\begin{eqnarray}
P_{out}&=&\mathrm{Pr}\left\{\frac{\mathbf{a}_e^H\mathbf{h}_n\mathbf{h}_n^H\mathbf{a}_e}{\mathbf{a}_e^H\mathbf{H}\mathbf{V}\mathbf{H}^H\mathbf{a}_e+\sigma_{w}^2}<\frac{P_{n|n-1}-\epsilon}{\epsilon P_{n|n-1}}\right\}\nonumber\\
&=&\mathrm{Pr}\left\{\mathbf{a}_e^H\mathbf{h}_n\mathbf{h}_n^H\mathbf{a}_e-\beta\mathbf{a}_e^H\mathbf{H}\mathbf{V}\mathbf{H}^H\mathbf{a}_e<\beta\sigma_{w}^2\right\}\nonumber\\
&=&\mathrm{Pr}\left\{\tilde{\mathbf{h}}_n^H\left(\mathbf{M}\tilde{\mathbf{a}}_e\tilde{\mathbf{a}}_e^H\mathbf{M}-\beta\mathbf{Q}\right)\tilde{\mathbf{h}}_n\le \frac{\beta\sigma_w^2}{P_T}\right\} \; \nonumber,
\end{eqnarray}where
\begin{eqnarray*}
\beta & = & \frac{P_{n|n-1}-\epsilon}{\epsilon P_{n|n-1}} \;,\;\tilde{\mathbf{a}}_e\;=\;\frac{1}{\sqrt{P_T}}\mathbf{a}_e\;,\\
\mathbf{M} & = & \mathrm{diag}\left\{\frac{1}{d_{1}^{\alpha}},\cdots, \frac{1}{d_{N}^{\alpha}}\right\}\;, \\ 
\mathbf{Q} & = & \mathrm{diag}\left\{\frac{\sigma_{v,1}^2}{N(\sigma_{\theta}^2+\sigma_{v,1}^2)d_{1}^{2\alpha}}
\;, \cdots, \; \frac{\sigma_{v,N}^2}{N(\sigma_{\theta}^2+\sigma_{v,N}^2)d_{N}^{2\alpha}}\right\}\;, \\ \tilde{\mathbf{h}}_n & = & [\tilde{h}_{1,n} \;, \cdots, \; \tilde{h}_{N,n}]^T \;.
\end{eqnarray*}

If we define $\mathbf{R}=\mathbf{M}\tilde{\mathbf{a}}_e\tilde{\mathbf{a}}_e^H\mathbf{M}-\beta\mathbf{Q}$, and label the eigenvalues of $\mathbf{R}$ as $\lambda_1,\cdots,\lambda_{N}$, then the random variable $\tilde{\mathbf{h}}_{n}^H\mathbf{R}\tilde{\mathbf{h}}_{n}$ can be viewed as the weighted sum of independent chi-squared random variables $\sum_{i=1}^{N}\lambda_i\chi_i(2)$.  From \cite{Al_Naffouri:2009}, we have
\begin{equation}\label{eq:cdf}
P_{out}=1-\sum_{i=1}^{N}\frac{\lambda_i^N}{\prod_{l\neq i}(\lambda_i-\lambda_l)}\frac{1}{|\lambda_i|}e^{-\frac{(P_{n|n-1}-\epsilon)\sigma_w^2}{\epsilon P_{n|n-1}P_T\lambda_i}}u(\lambda_i)\;,
\end{equation}where $u(\cdot)$ is the unit step function. Let $e_{1}\ge\cdots \ge e_{N}$ denote the eigenvalues of $\mathbf{Q}$, so that from Weyl's inequality \cite{So:1999} we have the following bounds for the $\lambda_i$:
\begin{eqnarray}
\tilde{\mathbf{a}}_{e}^H\mathbf{M}^2\tilde{\mathbf{a}}_{e}-\beta e_{1}\le &\!\!\!\lambda_{1}\!\!\!&\le \tilde{\mathbf{a}}_{e}^H\mathbf{M}^2\tilde{\mathbf{a}}_e-\beta e_{N}\;,\label{eq:lambda1}\\
-\beta e_{N-i+1}\le&\!\!\!\lambda_{i}\!\!\!&\le -\beta e_{N-i+2}\;,\quad 2\le i\le N\;\label{eq:lambda2}, 
\end{eqnarray}where $\tilde{\mathbf{a}}_{e}^H\mathbf{M}^2\tilde{\mathbf{a}}_{e}=\sum_{i=1}^{N}\frac{1}{N(\sigma_{\theta}^2+\sigma_{v,i}^2)d_{i}^{2\alpha}}$\;. From (\ref{eq:lambda1}), when $\beta$ is large, $\lambda_1$ is negative, and when $\beta$ is small enough, $\lambda_1$ is positive. Meanwhile, since all the eigenvalues of $\mathbf{Q}$ is positive, then according to (\ref{eq:lambda2}) we have that  $\lambda_{i}<0$ for $ 2\le i\le N$. Since only $\lambda_1$ can be positive, equation~(\ref{eq:cdf}) can be simplified as
\begin{equation}\label{eq:pouteq}
P_{out}=\left\{\begin{array}{lr}
1-\frac{\lambda_1^{N-1}}{\prod_{l\neq 1}(\lambda_1-\lambda_l)}e^{-\frac{(P_{n|n-1}-\epsilon)\sigma_w^2}{\epsilon P_{n|n-1}P_T\lambda_1}}&\textrm{$\lambda_1>0$}\\
1 &\textrm{$\lambda_1\le 0$}\;.
\end{array} \right.
\end{equation} From (\ref{eq:pouteq}), when the threshold $\epsilon$ is too small, $\beta=\frac{P_{n|n-1}-\epsilon}{\epsilon P_{n|n-1}}$ will be very large and $\lambda_1\le0$, then the outage probability $P_{out}$ equals 1, which means the MSE $P_{n|n}$ is larger than $\epsilon$ for every channel realization $\mathbf{h}_n$. For $P_T\to\infty$, the outage probability converges to
\begin{equation}
P_{out}=\left\{\begin{array}{lr}
1-\frac{\lambda_1^{N-1}}{\prod_{l\neq 1}(\lambda_1-\lambda_l)}&\textrm{$\lambda_1>0$}\\
1 &\textrm{$\lambda_1\le 0$}\;.
\end{array} \right.\nonumber
\end{equation}

\section{Simulation Results}\label{sec:nume}
To investigate the performance of the proposed optimization approaches, the
results of several simulation examples are described here. Unless otherwise indicated, the simulations are implemented with the following parameters: distance from the FC to the sensors $d_i$ is uniformly distributed over the interval $[2,8]$, path loss exponent is set to $\gamma=1$, the observation noise power $\sigma_{v,i}^2$ at the sensors is uniformly distributed over
$[0, 0.5]$, the power of the additive noise at the FC is set to $\sigma_w^2=0.5$, the parameter $\theta$ is assumed to satisfy $\sigma_{\theta}^2=1$, and the initial MSE is
given by $P_{0|-1}=0.5$.  The MSE shown in the plots is obtained by averaging over 300 realizations of $\mathbf{h}_n$.  Two different sum power constraints are considered in the simulations:
$P_T=300$ and $P_T=3000$. To fairly compare the results under sum and individual power constraints, we set $P_{T,i}=\frac{P_{T}}{N}$, which means that all sensors have the same maximum power when individual constraints are imposed.

Fig.~\ref{f1} plots the MSE as a function of the number of sensors in the network for both sum and individual power constraints.  The results demonstrate that compared with equal power allocation, the optimized power allocation significantly reduces the MSE; in fact, adding sensors
with equal power allocation actually increases the MSE, while the MSE
always decreases for the optimal methods.  The extra flexibility offered by
the global power constraint leads to better performance compared with
individual power constraints, but the difference in this case is not large. The lower bound on MSE in (\ref{eq:lb}) is also plotted to indicate the performance that that could be achieved with
$P_T\to\infty$.

Figs.~\ref{f3} and~\ref{f4} respectively examine sum and peak transmit powers required to achieve MSE values of 0.02, 0.04 and 0.1 for varying numbers of sensors.  As expected, individual power constraints lead to higher sum power requirements, while sum power constraints result in higher peak power.  Interestingly, the individual power constraints lead to roughly a doubling of the required total sum power to achieve the same MSE regardless of the number of sensors, whereas the increase in peak power for the sum constraint relative to individual power constraints grows with $N$, reaching a factor of~4 to~5 on average when $N=30$.  Fig.~\ref{f6} compares the minimum required sum transmit power to achieve various MSE values in~(\ref{eq:pmax}) with the approximate expression obtained in (\ref{eq:approxsolu}). When $\epsilon \ge 0.1$, the approximation is reasonably good even when $N$ is on the order of only 20 to 40.  The approximation is less accurate for tighter requirements on $\epsilon$, and requires a larger value of $N$ for the approximation to be valid.

The impact of the SNR at the FC on the sensor power allocation is illustrated in Fig.~\ref{f5} for a given channel realization and $N=30$ sensors.  The x-axis of each plot is ordered according to the channel gain of the sensors, which is shown in the upper left subfigure.  The upper right subfigure shows the variance of the measurement noise for each sensor, which for this example was uniformly drawn from the interval $[0.4,0.5]$ to better illustrate the effect of the channel gain.  The optimal power allocation for this scenario was found assuming both sum and individual power constraints under both low and high SNRs at the FC.  The middle subfigures show the power allocation for minimizing MSE assuming a low SNR case with $P_T=5$, while the bottom subfigures show the allocation for high SNR with $P_T=1000$.  Note that, as predicted by~(\ref{eq:optimal2a}), the power allocated to the sensors under the sum power constraint for low SNR tends to grow with the channel gain, while as predicted by~(\ref{eq:optimal2}), the allocated power is reduced with increasing sensor gain.  The explanation for the different behavior at low and high SNR can be explained as follows: when the SNR is high, the measurement noise will dominate the estimation error at the FC, and the higher the channel gain, the more the measurement noise is amplified, so the sensor nodes with higher channel gains will be allocated less power.  When the SNR is low, the additive noise at the FC will dominate the estimation error, the effect of the measurement noise can be neglected, so the nodes with higher channel gains will be allocated more power to increase the power of the desired signal.  For individual power constraints, we see that all of the sensors transmit with a maximum power of $P_T/N=5/30$ at low SNR, while at high SNR only the sensors with small channel gains use maximum power (in this case $P_T/N=1000/30$), and the power allocated to sensors with large channel gains decreases, as with the sum power constraint.

Finally, in Fig.~\ref{f7}, we show that our analytical expression in~(\ref{eq:pouteq}) for the outage probability under equal power allocation closely follows the simulation results for various transmit power levels for a case with $N=10$ sensors.  While these outage probabilities represent upper bounds for the optimal (and generally unequal) transmission gains, we note that these bounds are not particularly tight.  The outage probabilities achieved by the optimal algorithms are typically much lower than predicted by~(\ref{eq:pouteq}).

\section{Conclusion}\label{sec:conc}
In this paper, we considered the problem of optimally allocating power in an
analog sensor network attempting to track a dynamic parameter via a coherent 
multiple access channel.  We analyzed problems with either constraints on 
power or constraints on achieved MSE, and we also examined cases involving 
global sum and individual sensor power constraints.  While prior work had
been published for minimizing MSE under a sum power constraint and minimizing
sum power under an MSE constraint, we were able to derive closed-form solutions
that were simpler and more direct.  Going beyond the prior work, we derived new asymptotic 
expressions for the transmission gains that illustrated their limiting behavior
for both low and high SNR at the fusion center, and we found a simple expression
for the required sum transmit power when the number of sensors is large.
Furthermore, we showed how to minimize MSE under individual power constraints,
or minimize peak sensor power under MSE constraints, cases that had not been previously 
considered.  In particular, we demonstrated that solutions to these problems could
be found by solving a rank-relaxed SDP using standard convex optimization methods.
Finally, we derived an exact expression for the MSE outage probability for the
special case where the sensors transmit with equal power, and presented a number
of simulation results that confirmed our analysis and the performance of the
proposed algorithms. 

\appendices
\section{Proof of Theorem 2}
\begin{proof}
\iffalse
Note that to find the optimal power coefficient $\mathbf{a}_n^*$, we only need to calculate the optimal value of problem (\ref{eq:minspequ4}). \fi
Since $\mathbf{D}^{-\frac{1}{2}}\mathbf{E}_n\mathbf{D}^{-\frac{1}{2}}$ is the sum of a rank-one and a diagonal matrix, we have the following bounds for $\lambda_{\max}\{\mathbf{D}^{-\frac{1}{2}}\mathbf{E}_n\mathbf{D}^{-\frac{1}{2}}\}$:
\begin{eqnarray}
\lambda_{\max}\{\mathbf{D}^{-\frac{1}{2}}\mathbf{E}_n\mathbf{D}^{-\frac{1}{2}}\}&<&
P_{n|n-1}\mathbf{h}_n^H\mathbf{D}^{-1}\mathbf{h}_n-\left(\frac{P_{n|n-1}}{\epsilon}-1\right)
\min_{i}\left\{\frac{|h_{i,n}|^2\sigma_{v,i}^2}{\sigma_{\theta}^2+\sigma_{v,i}^2}\right\}\nonumber\\
&=& P_{n|n-1}\mathbf{h}_n^H\mathbf{D}^{-1}\mathbf{h}_n-\xi  \label{eq:highbound}\;,\\
\lambda_{\max}\{\mathbf{D}^{-\frac{1}{2}}\mathbf{E}_n\mathbf{D}^{-\frac{1}{2}}\}&>&
P_{n|n-1}\mathbf{h}_n^H\mathbf{D}^{-1}\mathbf{h}_n-\left(\frac{P_{n|n-1}}{\epsilon}-1\right)
\max_{i}\left\{\frac{|h_{i,n}|^2\sigma_{v,i}^2}{\sigma_{\theta}^2+\sigma_{v,i}^2}\right\}\nonumber\\
&=& P_{n|n-1}\mathbf{h}_n^H\mathbf{D}^{-1}\mathbf{h}_n(1-\zeta) \; ,\label{eq:lowbound}
\end{eqnarray}
where we define
\begin{eqnarray*}
\xi & = & \left(\frac{P_{n|n-1}}{\epsilon}-1\right)\min_{i}\left\{\frac{|h_{i,n}|^2\sigma_{v,i}^2}
{\sigma_{\theta}^2+\sigma_{v,i}^2}\right\} \; ,\\
\zeta & = & \frac{\left(\frac{P_{n|n-1}}{\epsilon}-1\right)\max_{i}\left\{\frac{|h_{i,n}|^2\sigma_{v,i}^2}
{\sigma_{\theta}^2+\sigma_{v,i}^2}\right\}}{P_{n|n-1}\mathbf{h}_n^H\mathbf{D}^{-1}\mathbf{h}_n}\;.
\end{eqnarray*}

For any positive constant $\nu$, we have
\begin{eqnarray}
\textrm{Pr}\left\{\xi\ge \nu\right\}&\le& \textrm{Pr}\left\{\eta\min_{i}\left\{|\tilde{h}_{i,n}|^2\right\}\ge \tilde{\nu}\right\}\nonumber\\
&=&\textrm{Pr}\left\{\min_{i}\left\{|\tilde{h}_{i,n}|^2\right\}\ge \frac{\tilde{\nu}}{\eta}\right\}\nonumber\\
&=&\left(1-\textrm{Pr}\left\{|\tilde{h}_{i,n}|^2\le \frac{\tilde{\nu}}{\eta}\right\}\right)^N\nonumber\\
&\overset{(b)}{=}&e^{-\frac{N\tilde{\nu}}{\eta}} \; \nonumber,
\end{eqnarray}
where
\begin{eqnarray*} \tilde{\nu} & = & \frac{\nu}{\frac{P_{n|n-1}}{\epsilon}-1}\;, \\ \eta & = & \max_{i}\left\{\frac{\sigma_{v,i}^2}{(\sigma_{v,i}^2+\sigma_{\theta}^2)d_{i,n}^\gamma}\right\} \; ,
\end{eqnarray*}
and $(b)$ is due to the fact that $2|\tilde{h}_{i,n}|^2$ is a chi-square random variable with degree 2. When $N\to\infty$, we have
\begin{eqnarray}\label{eq:conp1}
\lim_{N\to\infty}\textrm{Pr}\left\{\xi\ge \nu\right\}&\le& \lim_{N\to\infty}e^{-\frac{N\tilde{\nu}}{\eta}}\\
&=&0 \; ,\nonumber
\end{eqnarray}
and thus $\xi$ converges to $0$ in probability.

From the definition of $\zeta$,
\begin{eqnarray}
\zeta&=&\left(\frac{P_{n|n-1}-\epsilon}{P_{n|n-1}\epsilon}\right)\max_{i}\left\{\frac{\frac{|h_{i,n}|^2
\sigma_{v,i}^2}{\sigma_{\theta}^2+\sigma_{v,i}^2}}{\sum_{k=1}^{N}\frac{|h_{k,n}|^2}{\sigma_{\theta}^2+
\sigma_{v,k}^2}}\right\}\nonumber\\
&<&\left(\frac{P_{n|n-1}-\epsilon}{P_{n|n-1}\epsilon}\right)\max_{i}\left\{\frac{\frac{|h_{i,n}|^2
\sigma_{v,i}^2}{\sigma_{\theta}^2+\sigma_{v,i}^2}}{\sum_{k=1,k\neq i}^{N}\frac{|h_{k,n}|^2}{\sigma_{\theta}^2+\sigma_{v,k}^2}}\right\}\nonumber\\
&<&\tau\max_{i}\left\{\frac{|\tilde{h}_{i,n}|^2}{\sum_{k=1,k\neq i}^{N}|\tilde{h}_{k,n}|^2}\right\},\nonumber
\end{eqnarray}
where
\begin{equation*}
\tau=\left(\frac{P_{n|n-1}-\epsilon}{P_{n|n-1}\epsilon}\right)
\frac{\max_i\left\{\frac{\sigma_{v,i}^2}{(\sigma_{\theta}^2+\sigma_{v,i}^2)d_{i,n}^{\gamma}}\right\}}
{\min_i\left\{\frac{1}{(\sigma_{\theta}^2+\sigma_{v,i}^2)d_{i,n}^{\gamma}}\right\}} \; .
\end{equation*}
For any positive constant $\mu$, we have
\begin{eqnarray}\label{eq:maxdistr}
\textrm{Pr}\left\{\zeta\ge\mu\right\}&=&1-\textrm{Pr}\left\{\zeta\le\mu\right\}\nonumber\\
&\le&1-\textrm{Pr}\left\{\max_{i}\left\{\frac{|\tilde{h}_{i,n}|^2}{\sum_{k=1,k\neq i}^{N}|\tilde{h}_{k,n}|^2}\right\}\le\tilde{\mu}\right\}\nonumber\\
\iffalse
&=&1-\textrm{Pr}\left\{\max_i\left\{\frac{|\tilde{h}_{i,n}|^2}{\sum_{k=1,k\neq i}^{N}|\tilde{h}_{k,n}|^2}\right\}\le \tilde{\mu}\right\}\nonumber\\\fi
&=&1-\left(\textrm{Pr}\left\{\frac{|\tilde{h}_{i,n}|^2}{\sum_{k=1,k\neq i}^{N}|\tilde{h}_{k,n}|^2}\le \tilde{\mu}\right\}\right)^N\nonumber\\
&=&1-\left(\textrm{Pr}\left\{\frac{\sum_{k=1,k\neq i}^{N}|\tilde{h}_{k,n}|^2}{|\tilde{h}_{i,n}|^2}\ge \frac{1}{\tilde{\mu}}\right\}\right)^N\nonumber\\
&=&1-\left(\textrm{Pr}\left\{\frac{\sum_{k=1,k\neq i}^{N}|\tilde{h}_{k,n}|^2}{(N-1)|\tilde{h}_{i,n}|^2}\ge \frac{1}{(N-1)\tilde{\mu}}\right\}\right)^N,
\end{eqnarray}
where $\tilde{\mu}=\mu/\tau$.

In (\ref{eq:maxdistr}), the random variable $X=\frac{\sum_{k=1,k\neq i}^{N}|\tilde{h}_{k,n}|^2}{(N-1)|\tilde{h}_{i,n}|^2}$ has an \emph{F}-distribution with parameters $N-1$ and $2$. Thus, the cumulative density function of $X$ is given by \cite{Zhang:2011}
\begin{equation}
F(x)=\left(\frac{(N-1)x}{(N-1)x+1}\right)^{N-1} \; , \nonumber
\end{equation}
and thus
\begin{eqnarray}
\left(\textrm{Pr}\left\{X\ge \frac{1}{(N-1)\tilde{\mu}}\right\}\right)^N&=&\left(1-\textrm{Pr}\left\{X\le \frac{1}{(N-1)\tilde{\mu}}\right\}\right)^N\nonumber\\
&=&\left(1-\frac{1}{(1+\tilde{\mu})^{N-1}}\right)^N\nonumber\\
&=&\left(1-\frac{1}{(1+\tilde{\mu})^{N-1}}\right)^{(1+\tilde{\mu})^{N-1}\frac{N}{(1+\tilde{\mu})^{N-1}}}\nonumber.
\end{eqnarray}
Since $\tilde{\mu}>0$ and hence $\lim_{N\to\infty}(1+\tilde{\mu})^{N-1}=\infty$, we have
\begin{eqnarray}
\lim_{N\to\infty}\left(\textrm{Pr}\left\{X\ge \frac{1}{(N-1)\tilde{u}}\right\}\right)^N&=&\lim_{N\to\infty}\left(1-\frac{1}{(1+\tilde{\mu})^{N-1}}\right)^{(1+\tilde{u})^{N-1}\frac{N}{(1+\tilde{\mu})^{N-1}}}\nonumber\\
&=&\lim_{N\to\infty}e^{\frac{N}{(1+\tilde{u})^{N-1}}}.\nonumber
\end{eqnarray}
Furthermore,
\begin{equation}
\lim_{N\to\infty}\frac{N}{(1+\tilde{\mu})^{N-1}}=\lim_{N\to\infty}
\frac{1}{(1+\tilde{\mu})^{N-1}\ln(1+\tilde{\mu})}=0 \; ,\nonumber
\end{equation}
and thus
\begin{equation}\label{eq:maxpdf}
\lim_{N\to\infty}\left(\textrm{Pr}\left\{X\ge \frac{1}{(N-1)\tilde{\mu}}\right\}\right)^N=1\;.
\end{equation}
Substituting (\ref{eq:maxpdf}) into (\ref{eq:maxdistr}) yields
\begin{equation}\label{eq:conp2}
\lim_{N\to\infty}\textrm{Pr}\left\{\zeta\ge\mu\right\}=0 \; ,
\end{equation}
and we conclude that when $N\to\infty$, $\zeta$ converges to $0$ in probability.  The proof of the theorem is completed by substituting the results of (\ref{eq:highbound}), (\ref{eq:lowbound}), (\ref{eq:conp1}) and  (\ref{eq:conp2}) into (\ref{eq:pmax}).
\iffalse
The upper bound of optimal power is given by
\begin{eqnarray}
P_{\max}^{*}&=&\frac{\frac{P_{n|n-1}}{\epsilon}-1}{P_{n|n-1}\mathbf{h}_n^{H}\mathbf{D}^{-1}\mathbf{h}_n+o(P_{n|n-1}\mathbf{h}_n^{H}\mathbf{D}^{-1}\mathbf{h}_n)}\\
&=&\frac{P_{n|n-1}-\epsilon}{\epsilon P_{n|n-1}\mathbf{h}_n^H\mathbf{D}^{-1}\mathbf{h}_n}(1+o(1))
\end{eqnarray}

the eigenvector corresponding to $\lambda_{\max}{}$ is given by
\begin{equation}
\mathbf{y}^{*}=\mathbf{D}^{-\frac{1}{2}}\mathbf{h}_n
\end{equation}

then the approximate closed-form solution to problem (\ref{eq:minsp}) is
\begin{equation}
\mathbf{a}_n^{*}=\sqrt{\frac{(\frac{P_{n|n-1}}{\epsilon}-1)\sigma_w^2}{\mathbf{h}_n^H\mathbf{D}^{-1}\mathbf{E}_n\mathbf{E}^{-1}\mathbf{h}_n}}\mathbf{D}^{-1}\mathbf{h}_n
\end{equation}\fi
\iffalse
similarity with water filling
\fi
\end{proof}

\bibliography{reference}
\newpage

\begin{figure}
\centering
\includegraphics[height=3.5in, width=4.5in]{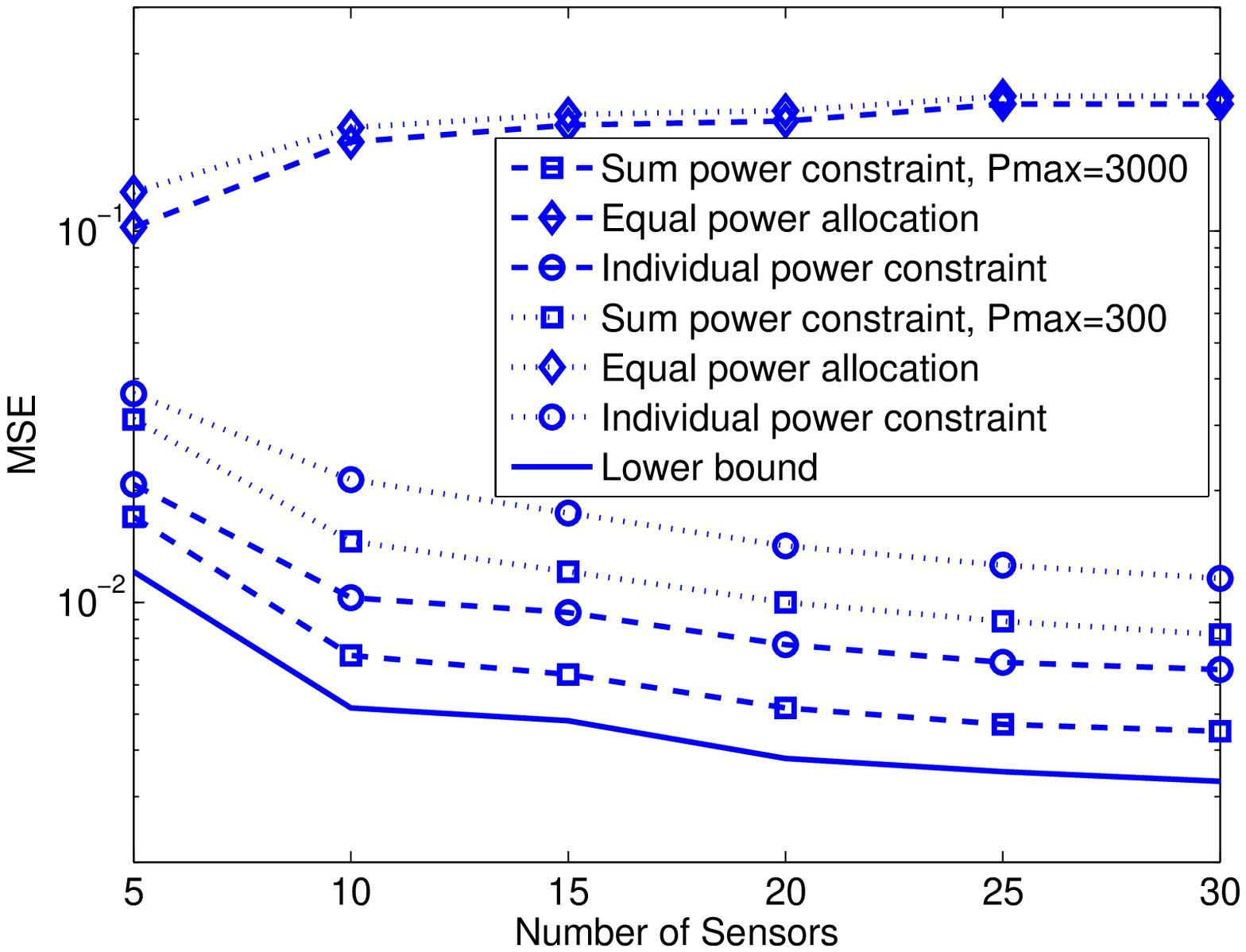}
\caption{MSE vs. number of sensors for $P_{\max}=300$ or $3000$.}\label{f1}
\end{figure}

\begin{figure}
\centering
\includegraphics[height=3.5in, width=4.5in]{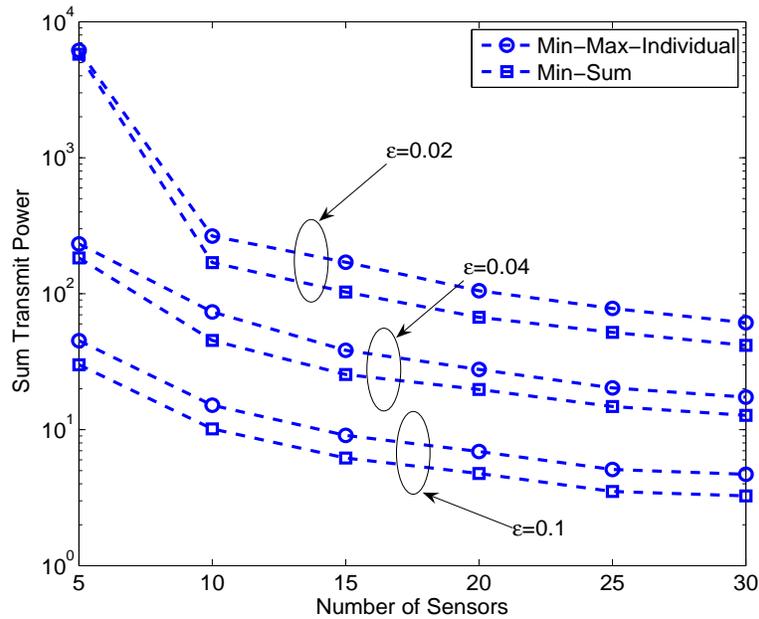}
\caption{Required sum transmit power vs. number of sensors for various MSE constraints.}\label{f3}
\end{figure}

\begin{figure}
\centering
\includegraphics[height=3.5in, width=4.5in]{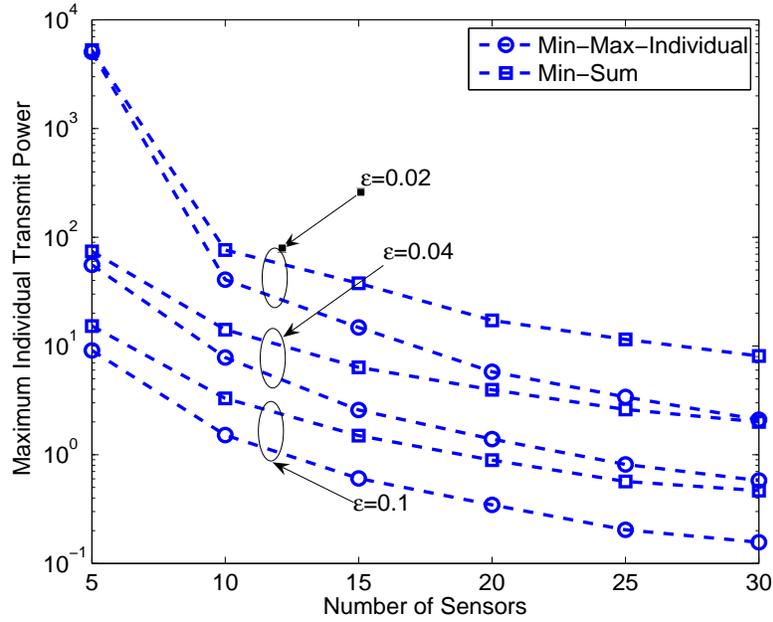}
\caption{Maximum individual transmit power vs. number of sensors for various MSE constraints.}\label{f4}
\end{figure}

\begin{figure}
\centering
\includegraphics[height=3.5in, width=4.5in]{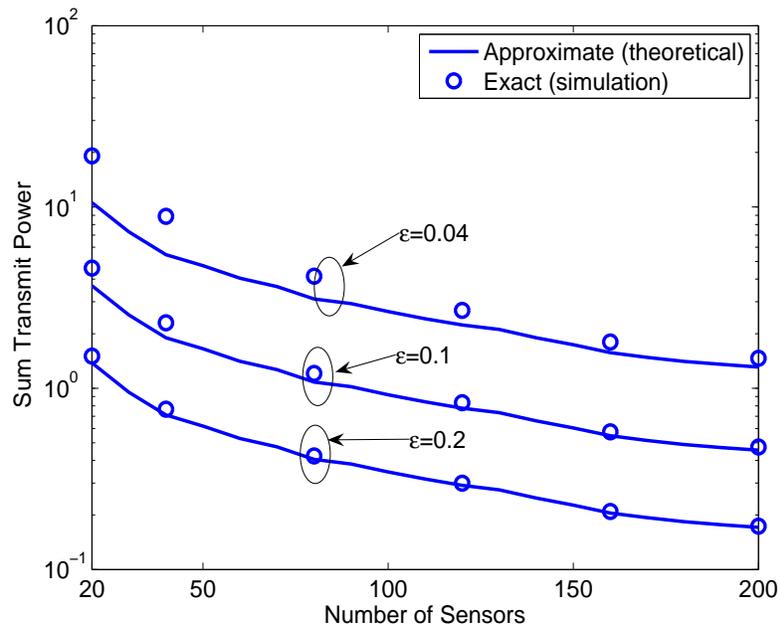}
\caption{Exact and approximate sum transmit power vs. number of sensors.}\label{f6}
\end{figure}

\begin{figure}
\centering
\includegraphics[height=6in, width=7in]{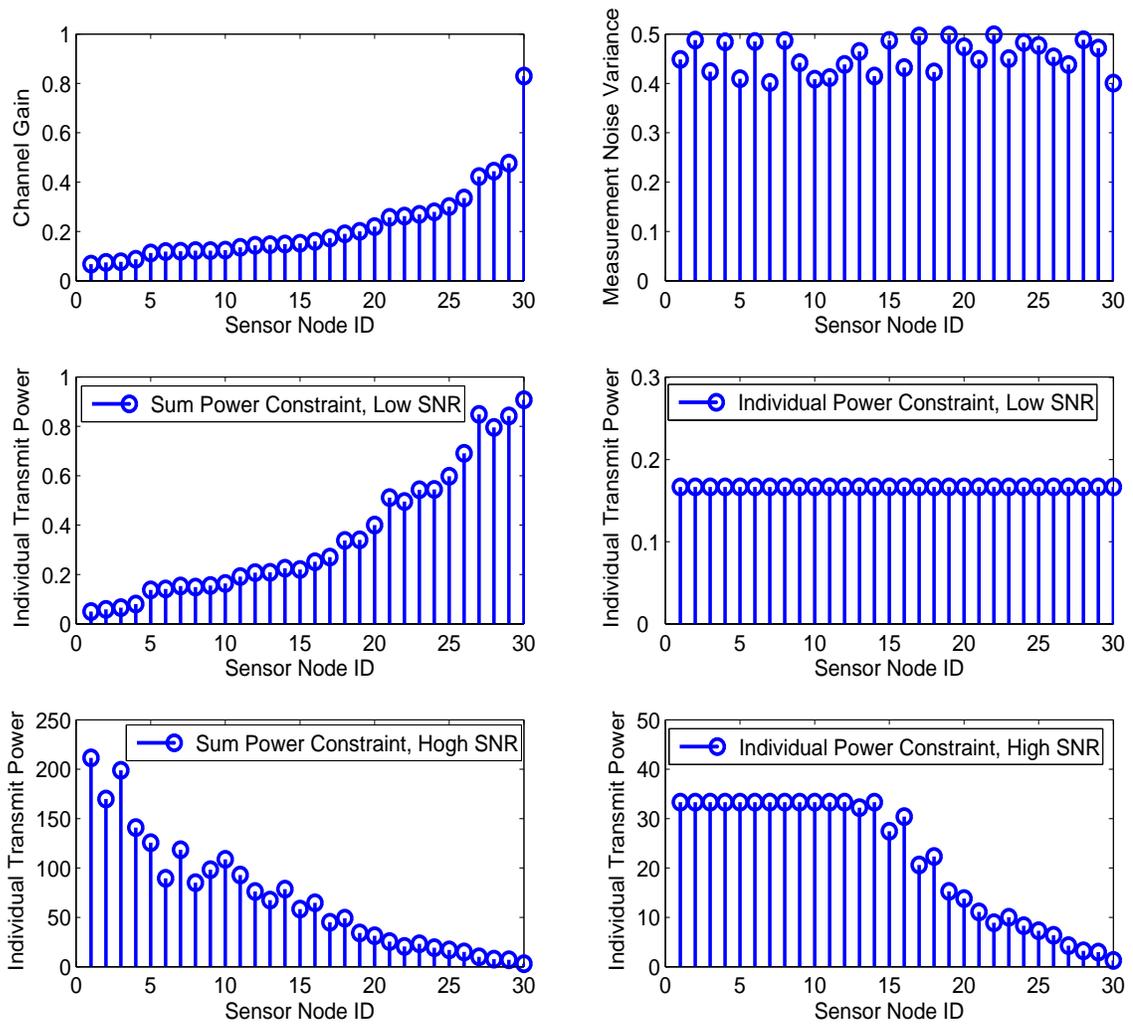}
\caption{Stem plot for the channel gain, measurement noise variance and the individual transmit power allocated to the sensor nodes. The x-axis denotes the sensor node ID and the sensor nodes are indexed according to their channel gain, in ascending order. For the high SNR case the total transmit power is set to $P_T=1000$ and for the low SNR case the total transmit power is $P_T=5$.}\label{f5}
\end{figure}

\begin{figure}
\centering
\includegraphics[height=3.5in, width=4.5in]{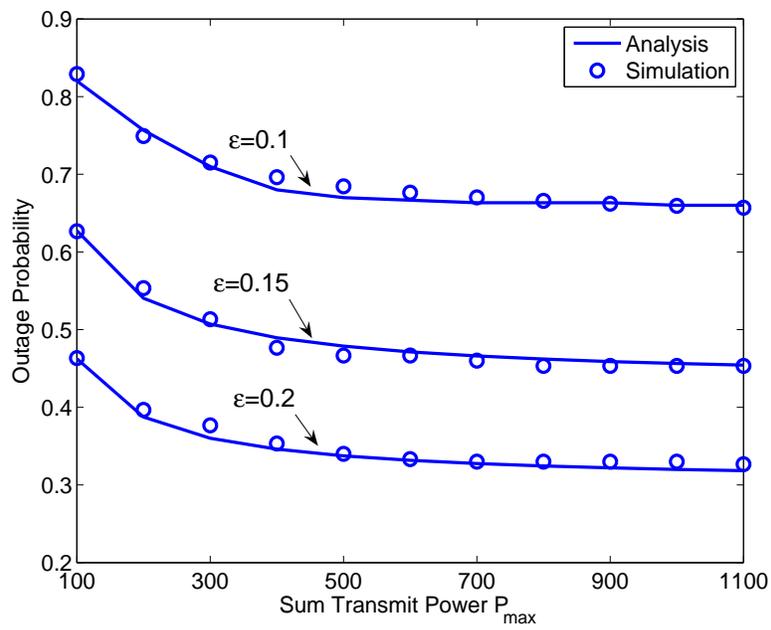}
\caption{MSE outage probability for equal power allocation vs.
sum transmit power for $N=10$ sensors.}\label{f7}
\end{figure}

\end{document}